\documentclass[proceedings]{aofa}
\usepackage{amsmath,amssymb}
\usepackage{graphics}
\usepackage{epsfig}  
\def\be{\begin{equation}}
\def\ee{\end{equation}}
\def\eps{\varepsilon}
\def\CA{{\cal A}}
\def\CB{{\cal B}}
\def\Res{\hbox{\rm Res}}
\def\bP{{\bf P}}
\def\bpi{{\boldsymbol\pi}}
\def\bzeta{{\boldsymbol\zeta}}
\def\bu{{\bf u}}
\def\bone{{\bf 1}}
\def\bI{{\bf I}}
\def\td{{\bar{\delta}}}

\def\bbe{{\bf e}}
\def\bT{{\bf T}}
\def\bD{{\bf D}}
\def\bR{{\bf R}}
\def\bM{{\bf M}}
\def\trace{\hbox{\rm trace}}

\def\parsec{\par\noindent}

\author{Philippe Jacquet\addressmark{1}
\and Wojciech Szpankowski\addressmark{2}\thanks{
W. Szpankowski is also with the Faculty of Electronics, Telecommunications and
Informatics,  Gda\'{n}sk University of Technology, Poland.
His work was supported by
NSF Center for Science of Information (CSoI) Grant CCF-0939370,
and in addition by NSF Grants CCF-1524312, and
NIH Grant 1U01CA198941-01, 
and the NCN grant, grant  UMO-2013/09/B/ST6/02258.
}}
\title[Average Size of a Suffix Tree for Markov Sources]
{Average Size of a Suffix Tree for Markov Sources\\ }
\address{\addressmark{1}Bell Labs, Alcatel-Lucent, France.\\
  \addressmark{2}Department of Computer Science, Purdue University, USA}
\keywords{Suffix tree, Markov sources, digital trees,  size, 
pattern matching, number of occurrences.}
\begin{document}
\maketitle

\begin{abstract}
We study a suffix tree built from 
a sequence generated by a Markovian source. 
Such sources are more realistic probabilistic models for text generation,
data compression, molecular applications, and so forth.
We prove that the average size of such a suffix tree 
is asymptotically equivalent to the average size of a trie built over 
$n$ independent sequences from the same Markovian source. This equivalence is
only known for memoryless sources. We then derive a formula for the size
of a trie under Markovian model to complete the analysis for suffix trees.
We accomplish our goal by applying
some novel techniques of analytic combinatorics on words also known as
analytic pattern matching.
\end{abstract}

\newtheorem{theorem}{Theorem}
\newtheorem{lemma}{Lemma}

\section{Introduction}

Suffix trees are the most popular data structures on words. They find
myriad of applications in
computer science and telecommunications, most notably
in algorithms on strings, data compressions (Lempel-Ziv'77 scheme), and codes.
Despite this, little is still known about their typical behaviors
for general probabilistic models
(see \cite{spa94,chauvin,ward}).

A suffix tree is a {\it trie} (a digital tree; see \cite{js-book}) 
built from the suffixes of a single string.  
In Figure~\ref{suffix-fig1} we show the suffix tree constructed for the
first four suffixes of the string $X=0101101110$.
\begin{figure}
\centerline{
\includegraphics[width=8cm]{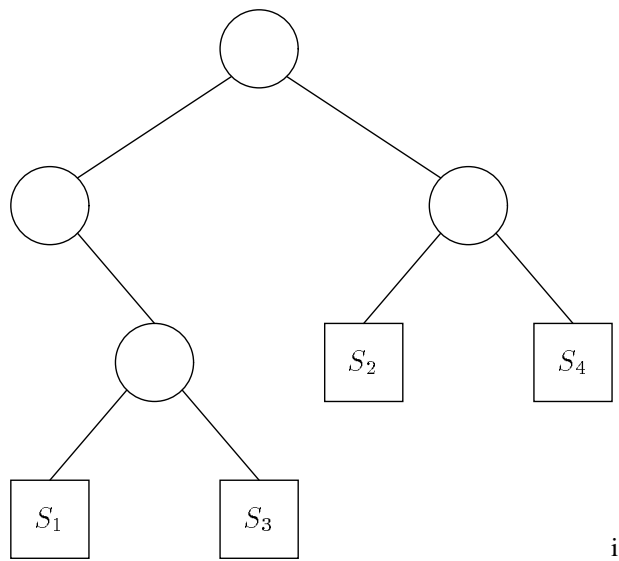}i}
\caption{Suffix tree built from the first five suffixes of
$X=0101101110$, {i.e.} $0101101110$,
$101101110$, $01101110$, $1101110$.}
\label{suffix-fig1}
\end{figure}
More precisely, we actually build a suffix tree on the first $n$ 
{\it infinite} suffixes of a string $X$
as shown in Figure~\ref{suffix-fig1}. We shall call it
simply a suffix tree which we study in this paper. 
Such a tree consists of internal (branching) nodes  and external node
storing the suffixes.
Our goal is to analyze the number of internal nodes called also the {\it size}
of a suffix tree built from a sequence $X$ generated by a Markov source.
We accomplish it by employing powerful techniques of 
analytic combinatorics on words known also as 
{\it analytic pattern matching} \cite{js-book}.

In recent years there has been a resurgence of interest in
algorithmic and combinatorial problems on words
due to a number of novel applications in
computer science, telecommunications, and most notably in molecular biology. 
A few possible applications are listed below. The reader is referred to
our recent book \cite{js-book} for more details.
In computer science and molecular biology many algorithms
depend on a solution to the following problem:
given a word $X$ and a set of arbitrary $b+1$ suffixes
$S_1$, ... , $S_{b+1}$ of $X$, what is the longest common prefix of
these suffixes.
In coding theory (e.g., prefix codes) one asks
for the shortest prefix of
a suffix $S_i$ which is not a prefix of
any other suffixes $S_j$, $1\leq j \leq n$ of a given sequence $X$
(cf. \cite{shields96}).
In data compression schemes,
the following problem is of prime interest:
for a given "data base" sequence of length $n$,
find the longest prefix of the $(n+1)$st suffix $S_{n+1}$ which is not
a prefix of any other suffixes $S_i$ ($1\leq i\leq n$) of the data base
sequence.
And last but not least, in molecular sequences comparison (e.g., finding
homology between DNA sequences),
one may search for the longest run of a given
motif, a unique sequence, the longest alignment, and 
the number of common subwords \cite{js-book}.
These, and several other problems on words,
can be efficiently solved and analyzed by a clever manipulation of
a data structure known as a {\it suffix tree}.
In literature other names have been also coined for this structure,
and among these we mention here position trees, subword trees,
directed acyclic graphs, {\it etc}. 

The extension of suffix tree analysis to Markov sources is
quite significant, especially when the suffix tree is used 
for natural languages.
Indeed, Markov sources of finite
memory approximate very well realistic texts. For example, the following
quote is generated by a memoryless source with the letter statistic of
the {\it Declaration of Independence}:
\begin{quote}
esdehTe,a; psseCed vcenseusirh vra f uetaiapgnuev n cosb mgffgfL
itbahhr nijue n S ueef,ru s,k smodpztrnno.eeteespfg mtet tr i aur oiyr
\end{quote}
which should be compared to the following quote generated by a
Markov source of order 3 trained on the same text:
\begin{quote}
We hat Government of Governments long that their right of abuses are
these rights, it, and or themselves and are disposed according Men, der.
\end{quote}

In this paper we analyze the average number of internal nodes (size) of a suffix
tree built from $n$ (infinite) suffixes of a string generated by a Markov
source with positive transition probabilities. We first
prove in Theorem~\ref{theo-main}
that the average size of a suffix tree under Markovian model is
asymptotically equivalent to the size of a {\it trie} that is built from
$n$ {\it independently} generated strings, each string emitted by 
the corresponding Markovian source.
To accomplish this, we study another quantity, namely the number of 
occurrences of a given pattern $w$ in a  
string of length $n$ generated by
a Markovian source. We use its properties 
to establish our asymptotic equivalence
between suffix trees and tries. Finally, we compare the average size of
suffix trees to trie size under Markovian model
(see Theorem~\ref{th-trie}), which 
-- to the best of our knowledge --
is only partially known \cite{cfv99}.

In fact, there is extensive literature on tries \cite{js-book} and very
scarce one on suffix trees. An analysis of the depth in a Markovian trie has been presented earlier in~\cite{js89}. A rigorous analysis of the depth of suffix tree 
was first presented in \cite{spa94} for memoryless sources, and then
extended in \cite{ward} to Markov sources. We should point out that
depth grows like $O(\log n)$ which makes the analysis manageable.
In fact, height and fillup level for suffix tree -- which are also of
logarithmic growth -- were analyzed in \cite{spa93} 
(see also \cite{chauvin,shields96}). But the average size grows like $O(n)$
and is harder to study. For memoryless sources it was analyzed in
\cite{rj} for tries and in \cite{spa94} for suffix trees. 
We also know that some parameters of suffix trees (e.g., profile) cannot
be inferred from tries, see \cite{gw}. Markov sources
add additional level of complications in the analysis of suffix trees
as well documented in \cite{chauvin}. In fact, the average size of
tries under general dynamic sources was analyzed in \cite{cfv99}, however,
specifications to Markov sources requires extra care, especially for the
so called rational Markov sources. 

The proof of the convergence of the average size of the suffix tree to the average size of the trie borrows many fundamental elements of the depth analysis in~\cite{ward}, for example the term $q_n(w)$ (see next section), but the extension of the depth analysis  to the size analysis require the introduction of a new term $d_n(w)$ which has non trivial properties. The analysis of average size of the trie in a Markovian model has been made by several author before but surprisingly we could not find a clear statement about the periodic case. This is the reason why we have to present a sketched proof here.


\section{Main Results}

We consider a stationary source generating a sequence of symbols drawn from a finite alphabet $\CA$. 

We first derive a formula for the average size of a suffix tree in terms of
the number of pattern occurrences.
Let $w$ be a word over $\CA$. We denote by $O_n(w)$ the number of occurrences of 
word $w$ in a sequence of length $n$ generated by a Markov source with
the transition matrix $\bP$. 
We observe \cite{spa94} that
the average size $s_n$ of a suffix tree built over a sequence of length 
$n$ is 
\be
\label{e1}
s_n=\sum_{w\in\CA^*}P(O_n(w)\ge 2).
\ee
In fact, (\ref{e1}) holds for any probabilistic source. We compare it to
the average size $t_n$ of trie built
over $n$ independent Markov sequences. If $N_n(w)$ is the number of words which begin with $w$ in a trie build with $n$ words, we have
\be
t_n=\sum_{w\in\CA^*}P(N_n(w)\ge 2).
\ee 
Let $P(w)$ be the probability of observing $w$ in a Markov sequence,  $N_n(w)$ is a Bernoulli $(n,P(w))$ and random variable $t_n$ can be written as
\be
t_n=\sum_{w\in\CA^*}1-(1-P(w))^n-nP(w)(1-P(w))^{n-1}
\ee

We specifically consider a Markovian source. 
We assume that the source is stationary and
ergodic. We will consider a Markovian process of order 1 with a 
positive transition matrix $\bP=[P(a|b)]_{a,b\in \CA}$.  
Extensions to higher order Markov  is possible
since a Markovian source of order $r$ is simply a Markovian source of order 1 
over the alphabet $\CA^r$. Notice that contrary to previous analysis we don't assume that $P(a|b)>0$ for all $(a,b)\in\CA^2$, since we allow that some transition may be forbiden and some other mandatory (while keeping the source ergodic). 

Our main result of the paper is formulated next,

\begin{theorem}
Consider a suffix tree built over $n$ suffixes of a sequence
of length $n$ generated by a Markov source
with a positive state transition matrix $\bP$.
There exists $\eps>0$ such that 
\be
\label{e-main}
s_n-t_n=O(n^{1-\eps})
\ee
for large $n$.
\label{theo-main}
\end{theorem}

In order to apply Theorem~\ref{theo-main} one needs to estimate the
average size of a trie under Markovian model. This seems to be unknown
except for some general dynamic sources \cite{cfv99}. In fact,
analysis of tries under Markovian sources is quite challenging
(see \cite{jst}). But we can offer the following result for the
average size of a trie under Markovian assumptions.
A sketch of the proof is presented in Section~\ref{sec-trie}.
 
\begin{theorem}
Consider a trie built over $n$ independent sequences generated by
a Markov source with positive transition probabilities.
For $(a,b,c)\in \CA^3$ define
\label{e-alpha}
\be
\alpha_{abc}=\log\left[\frac{P(a|b)P(c|a)}{P(c|b)}\right].
\ee
Then:\\
{\rm (aperiodic case)} If not all $\{\alpha_{abc}\}$ are commensurable, then
$$
t_n=\frac{n}{h} +o(n)
$$
where $h=-\sum_{a,b}\pi_a P(b|a) \log P(b|a)$ 
is the entropy rate of the underlying Markov source with $\pi_a$, $a\in \CA$,
denoting the stationary probability.
\parsec
{\rm (periodic case)} If all $\{\alpha_{abc}\}$ are commensurable, then
$$
t_n=\frac{n}{h}(1+Q(n))+O(n^{1-\eps})
$$
where $Q(n)$ is a periodic function and some $\eps>0$.
\label{th-trie}
\end{theorem}
\paragraph{Remark} 
We recall that a set of real numbers are commensurable 
(also known as "rationally related")
when their ratios are rational numbers.
We observe that if for all $(a,b)\in\CA^2$, the $\alpha_{abc}$ are commensurable for
one $c\in\CA$, then $\alpha_{abc}$ are commensurable for all values of $c$.Furthermore in the aperiodic case the $o(n)$ term can have a growth rate arbitrary close to order $n$, depending on source settings as shown in~\cite{frv} in the memoryless case.

In the rest of this section, we present a road map of the proof of (\ref{e-main}).
For this we will make use of ordinary generating functions. 
Let $w\in\CA^k$ be a word of length $k$. 
We also define $N_0(z,w)=\sum_{n>0} P(O_n(w)=0)z^n$ and 
$N_1(z,w)=\sum_{n>0}P(O_n(w)=1)z^n$  for $z\in\mathbb{C}$. 
We know from~\cite{js-book} that 
\begin{eqnarray*}
N_0(z,w)&=&\frac{S_w(z)}{D_w(z)}\\
N_1(z,w)&=&\frac{z^kP(w)}{D_w^2(z)}
\end{eqnarray*}
where $S_w(z)$ is the autocorrelation polynomial of word $w$ and 
$D_w(z)$ is defined as follows

\be
D_w(z)=S_w(z)(1-z)+z^kP(w)\left(1+F_{w}(z)(1-z)\right),
\ee
The memoryless case considers $F_w(z)=0$. The addition of a non zero  $F_w(z)$ is a significant change from the analysis in the memoryless case. In fact it captures the correlations between characters in the sequence and leads to non trivial developments. Here $F_{w}(z)$ for $w\in\CA^*-\{\eps\}$ is a function that depends on the Markov parameters of the source. It also depends only on the first and last character of $w$, say respectively $a$ and $b$ for $(a,b)\in\CA^2$ as described below.

Let $\bP$ be the transition matrix of the Markov source and  $\bpi$ be its 
stationary vector with $\pi_a$ its coefficient at symbol $a\in \CA$. 
The vector $\bone$ is the vector with all coefficients equal to 1 and 
$\bI$ is the identity matrix.
Assuming that $a\in\CA$ (resp. $b$) is the first (resp. last) symbol of $w$, 
we have \cite{rs98,js-book} 
\be
F_{w}(z)=\frac{1}{\pi_a}\left[(\bP-\bpi\otimes\bone)
\left(\bI-z(\bP+\bpi\otimes\bone)\right)^{-1}\right]_{b,a}
\ee
where $[{\bf A}]_{a,b}$ indicates the $(a,b)$ coefficient of the matrix
${\bf A}$, and $\otimes$ represents the tensor product.
An alternative way to express $F_w(z)$ is 
\be
\label{e-F}
F_{w}(z)=\frac{1}{\pi_a}\langle\bbe_a (\bP-\bpi\otimes\bone)
\left(\bI-z(\bP+\bpi\otimes\bone)\right)^{-1}\bbe_b\rangle
\ee
where $\bbe_c$ for $c\in\CA$ is the vector with a 1 at the 
position corresponding to symbol $c$ and all other coefficients are 0. 
Here $\langle {\bf x}, {\bf y}\rangle$ represents the scalar product of
${\bf x}$ and ${\bf y}$.

Let us define two important quantities:
\begin{eqnarray*}
d_n(w)&=&P(O_n(w)=0)-(1-P(w))^n, \\
q_n(w)&=&P(O_n(w)=1)-nP(w)(1-P(w))^{n-1}, 
\end{eqnarray*}
and their corresponding generating functions
\begin{eqnarray*}
\Delta_w(z)&=&\sum_{n>0}d_n(w)z^n\\
Q_w(z)&=&\sum_{n>0}q_n(w)z^n.
\end{eqnarray*}
Observe that $t_n-s_n=\sum_{w\in\CA^*} d_n(w)+q_n(w)$. 
Thus we need to estimate $d_n(w)$ and $q_n(w)$ for all $w\in \CA^*$.

We denote $\CB_k$  the set of words of length $k$ that do not overlap with 
themselves over more than $k/2$ symbols (see~\cite{js-book,spa94,ward} 
for more precise definition).  To be precise $w\in\CA^k-\CB_k$ if there exist $j>k/2$ and $v\in\CA^j$ and $(u_1,u_2)\in\CA^{k-j}$ such that $w=u_1v=vu_2$. This set plays a fundamental role in the analysis and it is already proven in~\cite{ward} that 
$$\sum_{w\in\CA^k-\CB_k}P(w)=O(\delta_1^k)$$ 
where $\delta_1$ is the largest coefficient in the Markovian transition 
matrix $\bP$. Since the authors of~\cite{ward} only consider strictly positive matrix $\bP$ we have $\delta_1<1$. Anyhow in the present paper we allow some coefficients to be equal to 1 or 0, as long the source is ergodic. Therefore $\delta_1$ may be equal to 1. To cope with this minor problem 
we define 
\begin{eqnarray*}
p&=&\exp\left(\lim\sup_{k,w\in\CA^k}\frac{\log P(w)}{k}\right)\\
q&=&\exp\left(\lim\inf_{k,w\in\CA^k,P(w)\neq 0}\frac{\log P(w)}{k}\right). 
\end{eqnarray*}
These quantities exist and are smaller than 1 since $\CA$ is a finite alphabet. 
From now we set $\delta=\sqrt{p}$ which replaces the parameter $\delta_1$ in the previous statements.

Now we are in the position to present two crucial lemmas, proved in the
next section, from which Theorem~\ref{theo-main} follows.

\begin{lemma}
There exist $\eps<1$ such that $\sum_{w\in\CA^*}q_n(w)=O(n^{\eps})$.
\label{lem1}
\end{lemma}

\begin{lemma}
There exists a sequence $R_n(w)$, for $w\in\CA^*$ such for 
all $1>\eps>0$ we have 
\begin{itemize}
\item{(i)} for $w\in\CB_k$: $d_n(w)=O((nP(w))^\eps k\delta^k)+R_n(w)$;
\item{(ii)} for $w\in\CA^k-\CB_k$: $d_n(w)=O((nP(w))^\eps)+R_n(w)$,
\end{itemize}
where $R_n(w)$ is such that $\sum_{w\in\CA^*}R_n(w)=O(1)$.
\label{lem2}
\end{lemma}
\paragraph*{Remark:} The sequence $d_n(w)$ is the main new element which makes the difference between the suffix tree depth analysis done in~\cite{ward} and the suffix tree size analysis. The later  was done in~\cite{js-book} for the memoryless case. The sequence $R_n(w)$ reflects the impact of the Markovian source on the analysis in particular is a consequence of the introduction of a non zero function $F_w(z)$.

\begin{proof}[of Theorem~\ref{theo-main}]
We already know via Lemma~\ref{lem1} that there exists 
$\eps<1$ such that $\sum_{w\in\CA^*}q_n(w)=O(n^\eps)$.
Let now $d^{(1)}_n=\sum_k\sum_{w\in\CB_k}(d_n(w)-R_n(w))$ and since for all 
$\eps>0$ observe that
$$
d_n^{(1)}=\sum_k\sum_{w\in\CB_k}O(n^\eps P^\eps(w)k\delta^k)=\sum_k O(n^\eps k(p^\eps\delta)^k),
$$ 
hence it converges for all $\eps>0$. 
Also let $d_n^{(2)}=\sum_k\sum_{w\in\CA^k-\CB_k}(d_n(w)-R_n(w))$.
Observe that
\begin{eqnarray*}
d^{(2)}_n&=&\sum_k\sum_{w\in\CA^k-\CB_k}O(n^\eps P^{\eps-1}(w)P(w))\\
&=&\sum_k\sum_{w\in\CA^k-\CB_k}O(n^\eps q^{(\eps-1)k}P(w))\\
&=&\sum_k O(n^\eps(\delta q^{\eps-1})^k),
\end{eqnarray*}
which converges for all $\eps$ such that $\delta q^{\eps-1}<1$ 
(take $\eps<1$ close enough to 1) and is $O(n^\eps)$. 
Finally $d_n^{(1)}+d_n^{(2)}+\sum_{w\in\CA^*}R_n(w)$ is also $O(n^\eps)$ 
for $\eps>0$ since $\sum_{w\in\CA^*}R_n(w)$ is finitely bounded. 
This completes the proof of Theorem~\ref{theo-main}.
\end{proof}

\section{Proof of Lemmas}

In this section we prove Lemma~\ref{lem1} and Lemma~\ref{lem2}.
In the proof of Lemma~\ref{lem1} we shall use some facts from \cite{ward},
however, our proof follows the pattern matching approach developed
in \cite{js-book}. 

\subsection{Proof of Lemma~\ref{lem1}}

The result is in fact already proven in ~\cite{ward}. 
Define
\be
Q_w(z)=P(w)\left(\frac{z^k}{D^2_w(z)}-\frac{z}{(1-(1-P(w))z)^2}\right).
\ee
In~\cite{ward} one defines $Q_n(1)=\frac{1}{n}\sum_{w\in\CA^*}q_n(w)$ and it is proven there that $Q_n(1)=O(n^{-\eps})$ for some $\eps>0$. 

\subsection{Proof of Lemma~\ref{lem2}}
First we have the following simple lemma.
The largest eigenvalue of $\bP$ is 1, let $\lambda_1,\lambda_2,\ldots$ 
be a sequence of other eigenvalues in the decreasing order of their modulus. 

\begin{lemma}
Uniformly for all $w\in\CA^*$ we find $F_w(z)=O(\frac{1}{1-|\lambda_1 z|})$.
\label{lemF}
\end{lemma}
\begin{proof}
By the spectral representation of $\bP$ we know that
$\bP=\bpi\otimes\bone+\sum_{i>0}\lambda_i\bu_i\otimes\bzeta_i$ 
where $\bu_i$ (resp. $\bzeta_i$) are the corresponding right 
(resp. left) eigenvectors. In fact we can introduce the matrices $\bD=\bpi\otimes\bone$ and $\bR=\sum_{i>0}\lambda_i\bu_i\otimes\bzeta_i$ whose spectral radius is $|\lambda_1|$ and satisfies the orthogonal property: $\bR\bD=\bD\bR=0$.

We have Let $\bM(z)=\bP-\bpi\otimes\bone)\left(\bI-z(\bP+\bpi\otimes\bone)\right)^{-1}$ we have $\bM(z)=\bR(\bone-z\bR)^{-1}$. Since $\bR^k=O(|\lambda_1|z)$ $\bR(\bI-z\bR)^{-1}$ is defined for all $z$ such 
that $|z|<\frac{1}{|\lambda_1|}$ and is $O(\frac{1}{1-|\lambda_1 z|})$, and so is $F_w(z)=[\bM(z)]_{a,b}$. 

\end{proof}

The next lemma is important.

\begin{lemma}
For $z$ such that $|\lambda_1 z|<1$ we have for all integers $k$
\be
\sum_{w\in\CA^{k+1}}P(w)F_w(z)=O(\lambda_1^k).
\ee
\label{lemFF}
\end{lemma}
\begin{proof}
The function $F_w(z)$ depends only on the first and last symbol of $w$. 
Considering a pair of symbols $(a,b)\in\CA^2$ the sum of the probabilities 
of the words of length $k+1$ starting with $a$ and ending with $b$, 
$\sum_{awb\in\CA^{k+1}}P(w)$, equals 
$\pi_a\langle\bbe_b\bP^k\bbe_a\rangle$.
Easy algebra leads to
\begin{eqnarray}
\sum_{w\in\CA^{k+1}}P(w)F_w(z)&=&\sum_{(a,b)\in\CA^2}
\langle\bbe_a\bM(z)\bbe_b\rangle\langle\bbe_b\bP^k\bbe_a\rangle \\
&=&\trace\left(\bM(z)\bP^k\right).
\end{eqnarray}
But since $\bP^k=\bD+\bR^k$ and $\bM(z)\bD=0$ and $\bR^k=O(|\lambda_1|^k)$, we conclude the proof

\end{proof}

We now follow a parallel approach to the approach developped in~\cite{ward} and in~\cite{spa94,js-book}. 

The generating function $\Delta_w(z)=\sum_{n\ge 0}d_n(w)z^n$  becomes
\begin{eqnarray}
\Delta_w(z)=\frac{P(w)z}{1-z}\left(\frac{1+(1-z)
F_{w}(z)}{D_w(z)}\right.
\left.-\frac{1}{1-z+P(w)z}\right).
\end{eqnarray}
We have 
$$d_n(w)=\frac{1}{2i\pi}\oint \Delta_w(z)\frac{dz}{z^{n+1}},
$$
integrated on any loop encircling the origin in the definition domain of 
$d_w(z)$. Extending the result in~\cite{spa94}, the authors of~\cite{ward} 
show that there exists $\rho>1$ such that the function $D_w(z)$ has a 
single root in the disk of radius $\rho$. 
Let $A_w$ be such a root. We have via the residue formula
\be
d_n(w)=\Res(\Delta_w(z),A_w)A_w^{-n}-(1-P(w))^n+d_n(w,\rho),
\ee
where $\Res(f(z),A)$ denotes the residue of function $f(z)$ on 
complex number $A$ and
\be
d_n(w,\rho)=\frac{1}{2i\pi}\oint_{|z|=\rho}\Delta_w(z)\frac{dz}{z^{n+1}}.
\ee

We have 
\be
\Res(\Delta_w(z),A_w)=\frac{P(w)\left(1+(1-A_w)F_{w}(A_w)\right)}{(1-A_w)C_w}
\ee
where $C_w=D'_w(A_w)$. But since $D_w(A_w)=0$ we can write
\be
\Res(\Delta_w(z),A_w)=-\frac{A_w^{-k}S_w(A_w)}{C_w}
\ee

We now consider asymptotic expansion of $A_w$ and $C_w$ as it is 
described in~\cite{js-book}, in Lemma 8.1.8 and Theorem 8.2.2. 
Although the expansions were presented for memoryless case, but for 
Markov source we simply replace $S_w(1)$ by $S_w(1)+P(w)F_w(1)$. 
We find 
\be
\begin{array}{rcl}
A_w&=&1+\frac{P(w)}{S_w(1)}\\
&&+P(w)^2\left(\frac{k-F_w(1)}{S^2_w(1)}-\frac{S_w'(1)}{S^3_w(1)} \right)+
O(P(w)^3)\\
C_w&=&-S_w(1)+P(w)\left(k-F_w(1)-2\frac{S_w'(1)}{S_w(1)}\right)\\
&&+O(P(w)^2)
\end{array}
\ee
Notice that these expansions in the Markov model first appeared in~\cite{ward}.

 From now follow the proof of Theorem  8.2.2 
in~\cite{js-book}. We define the function 
\be
\delta_w(x)=\frac{A_w^{-k}S_w(A_w)}{C_w}A_w^{-x}-(1-P(w))^{x}.
\ee
More precisely we define the function 
$$\td_w(x)=\delta_w(x)-\delta_w(0)e^{-x}$$ 
which has a Mellin transform $\delta_w^*(s)\Gamma(s)=\int_0^\infty\td_w(x)x^{s-1}dx$ 
defined for all $\Re(s)\in(-1,0)$ with
\begin{eqnarray}
\delta_w^*(s)=\frac{A_w^{-k}S_w(A_w}{C_w}\left[(\log A_w)^{-s}-1\right]
+1-\left[-\log(1-P(w))\right]^{-s}.
\end{eqnarray} 
When $w\in\CB_k$ with the expansion of $A_w$ and since 
$S_w(1)=1+O(\delta^k)$ and $S'_w(1)=O(k\delta^k)$, we find that 
similarly as shown in~\cite{js-book}
\be
\delta_w^*(s)=O(|s|k\delta^k)P(w)^{1-s}.
\ee
Therefore, by the reverse Mellin transform, for all $1>\eps>0$:
\begin{eqnarray}
\td(n,w)&=&\frac{1}{2i\pi}\int_{-\eps-i\infty}^{-\eps+i\infty}\delta_w^*(s)
\Gamma(s)n^{-s}ds\nonumber\\
&=&O(n^{1-\eps}P(w)^{1-\eps}k\delta^k)
\end{eqnarray}
When $w\in\CA^k-\CB_k$ we don't have the $S_w(1)=1+O(\delta^k)$. But it is 
shown in~\cite{ward} that there exists $\alpha>0$ such that for all 
$w\in\CA^*$: $S_w(z)>\alpha$ for all $z$ such that $|z|\le\rho$. 
Therefore we get 
$$\td(n,w)=O(n^{1-\eps}P(w)^{1-\eps}).$$  

We set
\be
R_n(w)=d_w(0)e^{-n}+d_n(w,\rho).
\ee
We first investigate the quantity $d_w(0)$. We need to prove that 
$\sum_{w\in\CA^*}d_w(0)$ converges. For this, 
noticing that 
$$S_w(A_w)=S_w(1)+\frac{P(w)}{S_w(1)}S_w'(1)+O(P(w)^2)$$
we obtain
\be
-\frac{A_w^{-k}S_w(A_w)}{C_w}=1-\frac{P(w)}{S_w(1)}
\left(F_w(1)+\frac{S_w'(1)}{S_w(1)}\right)+O(P(w)^2).
\ee
Thus
\be
d_w(0)=-\frac{P(w)}{S_w(1)}\left(F_w(1)+\frac{S_w'(1)}{S_w(1)}\right)+O(P(w))^2).
\ee
Without the term $F_w(1)$ we would have the same expression as 
in~\cite{js-book} whose sum over $w\in\CA^*$ converges. Therefore 
we need to prove that the sum $\sum_{w\in\CA^*}\frac{P(w)}{S_w(1)}F_w(1)$ 
converges. It is clear that the sum 
$$\sum_k\sum_{w\in\CA^k-\CB_k}\frac{P(w)}{S_w(1)}F_w(1)
$$ 
converges since 
$$\sum_{w\in\CA^k-\CB_k}P(w)=O(\delta^k)$$ 
and 
$F_w(1)$ is uniformly bounded. Now we consider the other part 
$$\sum_k\sum_{w\in\CB_k}\frac{P(w)}{S_w(1)}F_w(1).$$ 
We know that $S_w(1)=1+O(\delta^k)$, therefore
\be
\sum_{w\in\CB_k}\frac{P(w)}{S_w(1)}F_w(1)=\sum_{w\in\CB^k}P(w)F_w(1)+O(\delta^k).
\ee
But 
$$\sum_{w\in\CB^k}P(w)F_w(1)=\sum_{w\in\CA^k}P(w)F_w(1)+O(\delta^k),$$ 
and we know by Lemma~\ref{lemFF} that 
$\sum_{w\in\CA^k}P(w)F_w(1)=O(\lambda_1^k)$. 
Thus the sum $\sum_k\sum_{w\in\CA^k}\frac{P(w)}{S_w(1)}F_w(1)$ converges. 

The second and last effort concentrates on the term $d_n(w,\rho)$. 
We proceed as in the proof of Theorem 8.2.2 in~\cite{js-book}. 
We first have $d_n(w,\rho)=O(P(w)\rho^{-n})$ 
which is $O(n^\eps P(w)^\eps)$ without any condition on $w$. 
The issue is now to work on $w\in\CB_k$. In this case we have 
$S_w(z)=1+O(\delta^k)$ and therefore 
\begin{eqnarray}
d_n(w,\rho)&=&\frac{1}{2i\pi}\oint \frac{P(w)}{1-z}\left(\frac{1}{D_w(z)}
-\frac{1}{1-z+zP(w)}\right)\frac{dz}{z^{n+1}}\nonumber\\
&&+\frac{1}{2i\pi}\oint P(w)\frac{F_w(z)}{D_w(z)}\frac{dz}{z^{n+1}}.
\end{eqnarray}
We notice that the function 
$$ 
\frac{P(w)}{1-z}\left(\frac{1}{D_w(z)}-\frac{1}{1-z+zP(w)}\right)
$$ 
is $O(P(w)\delta^k)+O(P(w)^2)$, therefore the first integral is 
$O(P(w)\delta^k\rho^{-n})$. The second function 
$P(w)\frac{F_w(z)}{D_w(z)}$ is equal to $P(w)F_w(z)+O(P(w)\delta^k)$. 
We already know that $\sum_{w\in\CB_k}P(w)F_w(z)=O(\lambda_1^k)$, 
thus the series converges and  the lemma is proven. 

\section{Sketch of the Proof of Theorem~\ref{th-trie}}
\label{sec-trie} 

Let $a\in\CA$. We denote by $t_{a,n}$ the average size of a trie over 
$n$ independent Markovian sequences, all starting with the same symbol $a$. 
Then for $n\ge 2$
\be
\label{e-tn}
t_n=1+\sum_{a\in\CA}\sum_{k=0}^n {n\choose k}\pi_a^k(1-\pi_a)^{n-k}t_{a,k},
\ee
and similarly  for $b\in \CA$
\be
\label{e-tnb}
t_{n,b}=1+\sum_{a\in\CA}\sum_{k=0}^n {n\choose k}P(a|b)^k(1-P(a|b))^{n-k}t_{a,k},
\ee
where we recall $P(a|b)$ is the element of matrix $\bP$. 
Let $T(z)=\sum_nt_n\frac{z^n}{n!}e^{-z}$ and 
$T_a(z)=\sum_nt_{a,n}\frac{z^n}{n!}e^{-z}$ be the familiar Poisson transforms.
Using (\ref{e-tn}) and (\ref{e-tnb}) we find
\begin{eqnarray}
T(z)&=&1-(1+z)e^{-z}+\sum_{a\in\CA}T_a(\pi_a z),\\
T_b(z)&=&1-(1+z)e^{-z}+\sum_{a\in\CA}T_a(P(a|b) z).\
\end{eqnarray} 
Using dePoissonization arguments (see \cite{depo}) we shall
obtain $t_n=T(n)+O(\frac{1}{n}T(n))$.  Thus we need to study
$T(z)$ for large $z$ in a cone around the real axis. For this we apply
the Mellin transform that we describe next. In fact the convergence between the quantities $t_n$ and $T_n$ could also be derived by the application of the Rice method on the Mellin transform, since the later as an explicit form.

Let now $\bT(z)$ be the vector consisting of $T_a(z)$ for every $a\in\CA$.
It is not hard to see that its Mellin transform 
$$\bT^*(s)=\int_0^\infty\bT(z)z^{s-1}dz$$ 
is defined for $-1>\Re(s)>-2$ 
(since $\bT(z)=O(z^2)$ when $z\to 0$), and
\be
\bT^*(s)=-(1+s)\Gamma(s)\bone+\bP(s)\bT^*(s)
\ee
where $\bP(s)$ is the matrix consisting of $P(a|b)^{-s}$ if $P(a|b)>0$ 
and 0 otherwise. This identity leads to 
$$\bT^*(s)=-(1+s)\Gamma(s)(\bI-\bP(s))^{-1}\bone$$ 
where $\bI$ is the identity matrix. 
Similarly the Mellin transform $T^*(s)$ of $T(z)$ satisfies
\be
T^*(s)=-(1+s)\Gamma(s)+\langle\bpi(s),\bT^*(s)\rangle.
\ee
where $\bpi(s)$ is the vector composed of $\pi_a^{-s}$.

The inverse Mellin  transform of $T^*(s)$ is defined as
\be
T(n)=\frac{1}{2i\pi}\int_{c-i\infty}^{c+i\infty}T^*(s)n^{-s}ds, \ \ \ \ -1>c>-2.
\label{RMellin}
\ee
In order to find asymptotic behavior of $T(z)$ as $z\to \infty$ we need to
study the poles of $T^*(s)$ for $-2<\Re(s)$.
As discussed in \cite{jst,js-book} this is equivalent to analyzing the poles 
of $\bT^*(s)$. Since $(1+s)\Gamma(s)$ has no pole on $-2<\Re(s)<0$ we 
must consider poles of $(\bI-\bP(s))^{-1}$. In other words
(see \cite{jst,js-book}) we need to find $s$ for which the eigenvalue of largest modulus
$\lambda(s)$ of $\bP(s)$ is equal to 1. It is easy to see that $\lambda(-1)=1$ 
since $\bP(-1)=\bP$. The residue at $s=-1$ of 
$n^{-s}(\bI-\bP(s))^{-1}\bone$ is equal to $\frac{n}{h}\bone$ 
where $h$ is the entropy rate of the Markovian source. 

As explained in \cite{jst} in the periodic case there are multiple values of 
$s$ such that $\lambda(s)=1$ and $\Re(s)=-1$. Since these poles 
are regularly spaced 
on the axis $\Re(s)=0$, they contribute to the oscillating terms 
(function $Q$ in Theorem~\ref{th-trie}) in the asymptotic expansion of 
$t_n$. Furthermore, the location of zeros of $\lambda(s)=1$ in the
 periodic case tells us that there exists 
$\eps$ such that  $(\bI-\bP(s))$ has no pole for $-1<\Re(s)<-1+\eps$ 
leading to the error term $O(n^{1-\eps})$.

In the aperiodic case there is only one pole on the line $\Re(s)=-1$, 
thus the oscillating term disappears. 
However, zeros of $\lambda(s)=1$ can lie arbitrarily close to
the line $\Re(s)=1$, therefore the error term is just $o(n)$.

\end{document}